\documentclass[a4paper]{article}
\usepackage[a4paper,includeheadfoot,margin=2.54cm]{geometry}

\usepackage[T1]{fontenc}
\usepackage{amsthm}
\usepackage{amsmath}
\usepackage{amssymb}
\usepackage[affil-it]{authblk}
\usepackage{microtype}
\usepackage{cite}
\usepackage{enumerate} 
\usepackage{graphicx}

\usepackage[noabbrev,capitalise]{cleveref}

\usepackage{color}

\usepackage{etoolbox}
\usepackage{lmodern}
\makeatletter
\patchcmd{\@maketitle}{\LARGE \@title}{\fontsize{16}{19.2}\selectfont\@title}{}{}\makeatother

\theoremstyle{plain}
\newtheorem{theorem}{Theorem}[section]
\newtheorem{lemma}[theorem]{Lemma}
\newtheorem{observation}[theorem]{Observation}
\newtheorem{corollary}[theorem]{Corollary}

\newcommand{\eps}{\varepsilon}
\newcommand{\T}{\mathcal{T}}

\newcommand{\N}{\mathbb{N}}
\newcommand{\R}{\mathbb{R}}
\newcommand{\Rr}{\mathcal{R}}
\newcommand{\G}{\mathcal{G}}
\newcommand{\I}{\mathcal{I}}

\DeclareMathOperator{\DG}{DG}
\DeclareMathOperator{\true}{true}
\DeclareMathOperator{\lab}{\ell}

\DeclareMathOperator{\node}{node}
\DeclareMathOperator{\port}{port}
\DeclareMathOperator{\id}{id}
\DeclareMathOperator{\diam}{diam}

\title{Routing in Unit Disk Graphs without Dynamic Headers}

\author{Wolfgang Mulzer\footnote{
   Supported by ERC StG 757609.
} }
\author{Max Willert}

\affil{Institut f\"ur Informatik, Freie Universit\"at Berlin, Germany\\
\texttt{\{mulzer,willerma\}@inf.fu-berlin.de}}

\date{}
\begin{document}
\maketitle

\begin{abstract}
Let $V\subset\R^2$ be a set of $n$ sites in the plane. 
The \emph{unit disk graph}
$\DG(V)$ of $V$ is the graph with vertex set $V$ in which 
two sites $v$ and $w$ are adjacent if
and only if their Euclidean distance is at most $1$.

We develop a \emph{compact routing scheme} $\Rr$ for $\DG(V)$. 
The routing scheme $\Rr$ preprocesses
$\DG(V)$ by assigning a \emph{label} $\lab(v)$ to
every site $v$ in $V$. After that, for any two sites $s$ and $t$, 
the scheme $\Rr$ must be able to route a packet from $s$ to $t$ as 
follows: given the label of a \emph{current vertex}
$r$ (initially, $r=s$) and the label of the target vertex $t$, the 
scheme determines a neighbor $r'$ of $r$. Then, the packet 
is forwarded to $r'$, and  the process continues until the packet 
reaches its desired target $t$.
The resulting path between the source $s$ and the 
target $t$ is called the
\emph{routing path} of $s$ and $t$. 
The \emph{stretch} of the routing scheme is the maximum ratio of 
the total Euclidean
length of the routing path and of the shortest path in $\DG(V)$,
between any two sites $s, t \in V$. 

We show that for any given $\eps>0$, 
we can construct a routing scheme for $\DG(V)$ with diameter $D$ that achieves
stretch $1+\eps$ and label size
$O(\log D\log^3n/\log\log n)$ (the constant in the $O$-Notation depends on $\eps$).
In the past, several routing schemes for unit disk graphs 
have been proposed.
Our scheme is the first one to achieve poly-logarithmic 
label size and arbitrarily small stretch without storing 
any additional information in the packet.
\end{abstract}

\section{Introduction}
The \emph{routing problem} is a well-known problem in distributed graph algorithms
\cite{PelegUp89,GiordanoSt04}. We are given a graph $G$ and want to preprocces
it by assigning \emph{labels} to each node of $G$ such that the following
task can be solved: a data packet is located at a source node and has to be routed
to a target node. A routing scheme should
have several properties. First, routing must be \emph{local}: a node
can only use the label of the target node as well as its own local 
information
to compute
a neighbor to which the packet is sent next. Second, the routing should be \emph{efficient}:
the ratio of the routed path and the shortest path --- the \emph{stretch factor} --- should
be close to $1$.
Finally, the routing scheme should be \emph{compact}: the size of the 
labels (in bits) must be
small.

In the literature, one can find many different techniques and models for routing. A common
tool is the use of \emph{routing tables}. A routing table is a
sequence of bits
stored in a node. Typically, routing tables contain more information about the topology
of the graph and are different from labels. In this article, we 
do not use routing tables, but store all the information in the labels.
Moreover, many routing schemes use additional \emph{headers}. The header
contains mutable information and is stored in the data packet. Thus, the header
moves
with the data packet through
the graph. The usage of an additional header makes it possible to implement
recursive routing strategies or to remember information from past 
positions of the packet.

Furthermore, the literature distinguishes two types of input models. 
In the
\emph{fixed-port model}, the given graph already has a complete
list of \emph{ports} for each node $v$, i.e., a fixed numbering 
of the neighbors of $v$ used to identify the next hop of the 
packet.
In particular, it is not possible to renumber the ports. In contrast,
the \emph{designer-port model} allows us to assign arbitrary
port numbers during the preprocessing, see \cite{FraigniaudGa01,ThorupZw01,fraigniaud2002space}. 
Below, we will briefly discuss the advantages and disadvantages of these two models.

A trivial solution to solve the routing problem is to store the complete shortest path
tree in every label. Then it is easy to route the data packets along a shortest path.
However, such a routing scheme is not compact. Moreover, Peleg and Upfal~\cite{PelegUp89}
proved that in general graphs, any routing scheme that achieves a constant stretch factor
must store a polynomial number of bits for each node.

Nevertheless, there is a rich collection of routing schemes for general graphs
\cite{AbrahamGa11,AwerbuchBNLiPe90,Chechik13,Cowen01,
EilamGaPe03,RodittyTo15,RodittyTo16}.
For example, the scheme by Roditty and Tov~\cite{RodittyTo16} 
uses labels of size $mn^{O(1/\sqrt{\log n})}$ and
routes a packet from $s$ to $t$ on a path of length 
$O\big(k\Delta+m^{1/k}\big)$, where $\Delta$ is the 
shortest path distance between $s$ and $t$, $k > 2$ is any 
fixed integer, $n$ is the number of nodes, and $m$ is the number of
edges. Their routing scheme needs headers of poly-logarithmic size.

The lower bound result by Peleg and Upfal~\cite{PelegUp89} shows that it is hopeless
to find efficient routing schemes for general graphs that are compact as well,
meaning that at most a poly-logarithmic number of bits in the labels/tables are
necessary.
Thus, it is natural to investigate special interesting graph classes and to develop compact and
efficient routing schemes for them. For example, it is possible to
route in trees along a shortest path by using a poly-logarithmic number of bits in the
label~\cite{FraigniaudGa01,SantoroKh85,ThorupZw01}. 
Moreover, in planar graphs, for any fixed $\eps > 0$, we can find a routing scheme
that achieves the stretch factor $1 + \eps$. Again, the number of bits for the
labels is poly-logarithmic~\cite{Thorup04}. The same holds for visibility graphs
of simple polygons~\cite{BanyassadyCKMRR17}. Moreover, see \cite{AbrahamGGM06} for
different routing compact routing schemes in networks with low doubling dimension.

Our graph class of interest comes from the study of mobile and wireless networks.
These networks are usually modeled as \emph{unit disk graphs} \cite{clark1990unit}. Nodes
in this network are points in the plane and two nodes are connected if their distance
is at most one. This is equivalent to a disk intersection graph in which all disks
have diameter one.
For unit disk graphs there are known routing schemes. The first routing scheme
is by Kaplan et al.~\cite{KaplanMuRoSe18} and uses the fixed-port
model. They present a
routing scheme with stretch
$1+\eps$ and routing table size $O(\log^2n\log^2D)$, where $D$ is the diameter
of the given unit disk graph. Their routing is recursive and needs an additional header
of size $O(\log n\log D)$. The second routing scheme is due to Yan, Xiang, and Dragan~\cite{yan2012compact}. They present a routing scheme with label size $O(\log^2n)$ and
show that a data packet routes along a path of length a most $5\Delta+13$, where $\Delta$
is the length of the optimal path. The designer-port model is used.

Here, we present the first compact routing scheme that is headerless and
achieves stretch $1+\eps$. We obtain label size $O(\log D\log^3n/\log\log n)$.
\footnote{The constant in the $O$-Notation depends on $\eps$.}
We use the fixed-port model. In the conclusion,
we will discuss how our scheme compares to the other schemes.

\section{Preliminaries}

We explain our graph theoretic notation and 
discuss how the routing scheme can access the
input graph. Then, we provide a precise definition 
of our notion of a routing scheme and give some 
background on unit disk graphs.

We are given a \emph{simple} and \emph{undirected}, 
graph $G = (V, E)$ with $n$ vertices.
The edges are weighted by a non-negative
weight function $w \colon E \rightarrow \R_0^+$.
We write $d_G(s,t)$ for the (weighted) \emph{shortest 
path distance} between the vertices $s, t \in V$ and 
we omit the subscript $G$ if it follows from the context.
Throughout the whole article we assume that the graph
is \emph{connected}.

\subparagraph*{Graph Access Model.}
Let $\Sigma = \{0, 1\}$, and 
$[m] = \{0, 1, \dots, m\}$, for $m \in \N$. 
We explain how the routing scheme may access
the input graph $G = (V, E)$.
Every vertex $v \in V$ has an identifier
$v_{\id} \in\Sigma^+$ of length
$|v_{\id}| = \lceil\log n \rceil$.
We use the \emph{fixed-port model} 
\cite{FraigniaudGa01,ThorupZw01,fraigniaud2002space}.
In this model the port numbers are
assigned arbitrarily.
The neighbors of a vertex $v \in V$ are accessed
through \emph{ports}. 
More precisely, 
there is a partial function
$\node: V \times [n - 1] \rightarrow V$,
that assigns to every vertex $v \in V$ and 
to every port number $p \in [n - 1]$ the 
neighbor $w = \node(v, p)$ that can be reached
through the port $p$ at vertex $v$. For simplicity, 
we set $\node(v, 0) = v$, for all $v\in V$.
In our algorithms, we use \emph{broadcast functions} 
$\beta_v: \Sigma^+ \rightarrow [n]$,
for every vertex $v \in V$.
It is defined as follows:
\[
\beta_v (w_{\id}) = \left\{
\begin{array}{ll}
p, & \textrm{if }\node(v, p) = w, \textrm{ and} \\
n, & \textrm{otherwise.} 
\end{array}
\right.
\]
The broadcast functions can be implemented with the
$\node$ functions as follows:
ask all neighbors of a node $v$ whether they 
have the identifier $w_{\id}$. If there is one, 
then this neighbor will answer on the corresponding port $p$. 
Otherwise, we output $n$.

Other authors also use the
\emph{designer-port model} \cite{FraigniaudGa01,ThorupZw01,yan2012compact}. In 
this model, the routing scheme can determine the 
assignment of port numbers to the incident edges of 
each vertex $v \in V$ during the preprocessing phase.
This additional power in the model can lead to more
efficient routing schemes \cite{FraigniaudGa01,ThorupZw01,yan2012compact,fraigniaud2002space}.
However, a routing scheme that uses the designer-port 
cannot easily be used as a building block for more
complicated routing schemes, since 
additional lookup tables become necessary in order to
store the assignments of the port numbers.

\subparagraph*{Routing Schemes.}
Let $\G$ be a graph class. A \emph{routing scheme} $\Rr$ for $\G$ 
consists of a family of \emph{labeling functions}
$\lab_G: V(G) \rightarrow \Sigma^+$, for each
$G \in \G$. The labeling function $\lab_G$ assigns 
a bit string $\lab_G(v)$ to every node $v$ of $G \in \G$. 
The label $\lab_G(v)$ serves as the address of the node 
$v \in V$ in $G$. 
In contrast to the identifier of a node, the label usually
contains the identifier, but some more information about the
topology of the graph $G$.
While the identifier is given as fixed input, the label is chosen
by the routing scheme during the preprocessing.
As before, we omit the index $G$ if the context is clear.
Furthermore, $\Rr$ has a \emph{routing function} 
$\sigma: \Sigma^+ \times \Sigma^+ \times \N^{\Sigma^+}\rightarrow \N$.
 The routing function $\sigma$ describes the
behavior of the routing scheme, as follows: assume
a data packet is located at a vertex $s \in V$ and must
be routed to a destination $t \in V$. Then, 
$\sigma(\lab(s),\lab(t),\beta_s)$ has to compute a port 
$p$ so that the next hop of the data packet is from 
$s$ to $\node(s, p)$. Now, let $v_0 = s$ and
$v_{i + 1} = \node(v_i, \sigma(\lab(v_i) ,\lab(t), \beta_{v_i}))$, 
for $i \geq 0$. The sequence $(v_i)_{i \in \N}$ is called
\emph{routing sequence}. The routing scheme $\Rr$ is \emph{correct}, for 
$G \in G$, if and only if for all distinct $s, t\in V(G)$, there
is a number $m(s, t) \in \N$ such that $v_{j} = t$, for all $j \geq m(s, t)$,
and $v_{j} \neq t$, for all  $j = 0, \dots, m(s, t) - 1$.
If $\Rr$ is correct for $G = (V, E)$,
then $\delta_G(s, t) = \sum_{i=1}^{m(s,t)}w(v_{i-1},v_i)$ is called
the \emph{routing length} between $s$ and $t$ (in $G$).
The \emph{stretch} of the routing scheme is the largest ratio 
$\delta_G(s, t)/d_G(s,t )$ over all distinct vertices $s, t \in V$.
The goal is to achieve a routing scheme that minimizes
the stretch factor as well as the number of bits stored in the labels.
Many routing schemes use additional \emph{headers} during the routing.
These headers as well as the target labels are stored in the data packet.
In contrast to the target label, the header might change while the packet
is routed through the graph. This gives additional power and makes it
possible to develop recursive routing schemes. However, we will not make
use of this technology.

\subparagraph*{Unit Disk Graphs.}
Our graph class of interest are the \emph{unit disk graphs}.
Let $V \subset \R^2$ be a set of $n$ points in the 
Euclidean plane. The unit disk graph $\DG(V)$ of $V$ has
vertex set $V$ and an edge between two vertices $v, w \in V$
if and only if the \emph{Euclidean distance} $|vw|$ is at most $1$,
see \cref{fig:unit-disk-graph}.
The weight of the edge $vw$ is $|vw|$.
Throughout, we will assume that $\DG(V)$ is connected, and 
we will use $D$ to denote the diameter
$\max_{u, v \in V} d(u, v)$ of $\DG(V)$. Clearly,
we have $D \leq n - 1$.

\begin{figure}
\begin{center}
\includegraphics[scale=1]{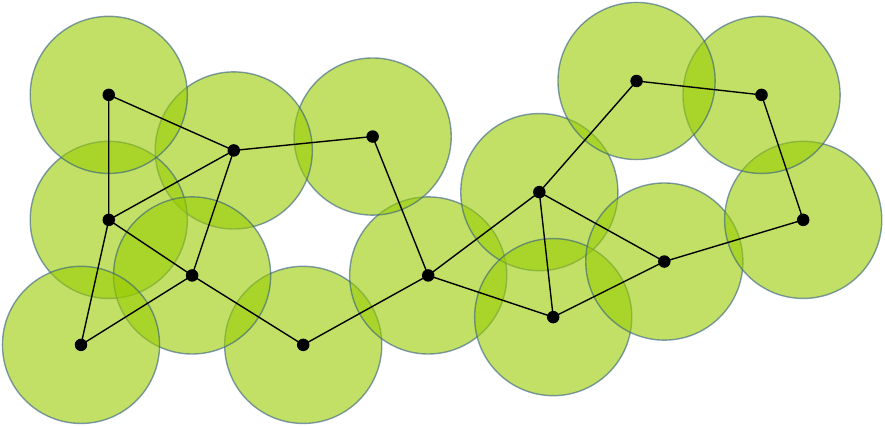}
\end{center}
\caption{The disks in the unit disk graph have diameter 1 and there is an edge between
two midpoints if and only if their corresponding disks intersect.}
\label[figure]{fig:unit-disk-graph}
\end{figure}

\section{Building Blocks}

In this section, we describe the building blocks
for our routing scheme. For this, we review some
simple routing schemes from the literature, and
we show how to obtain a new routing scheme for
unit disk graphs that achieves an 
\emph{additive} stretch. This later scheme 
is based on the data structure of Chan and Skrepetos.

\subsection{Simple Routing Schemes}

The first routing scheme is for trees. There are 
many different such schemes, based on similar ideas.
We would like to point out that some of these routing 
schemes can achieve label size $O(\log n)$, see\cite{FraigniaudGa01, ThorupZw01}.
However, 
these routing schemes work only in the designer-port model
and therefore are not useful as building blocks for 
more complex routing schemes,
especially if---as in our routing case---we need to 
be able to route in several subtrees of the input graph.\footnote{In 
fact, there is a lower bound that shows that label size $O(\log n)$
cannot be achieved in the fixed-port model \cite{fraigniaud2002space}.}
The following lemma is due to Fraigniaud and Gavoille \cite{FraigniaudGa01}
as well as Thorup and Zwick \cite{ThorupZw01}.

\begin{lemma}
\label[lemma]{lem:tree-routing}
Let $T$ be an $n$-vertex tree with arbitrary edge weights. 
There is a routing scheme for $T$ with
label size $O(\log^2 n/\log\log n)$ whose routing function 
$\sigma_\textup{tree}$ sends a data packet along
a shortest path, for any pair of vertices.
\end{lemma}

The second routing scheme is efficient for unit disk graphs with 
small diameter. The idea of the scheme was first
described by Kaplan et al. \cite{KaplanMuRoSe18}. They use the following 
lemma, which is based on a method by Gao and Zhang \cite{gao2005well}.

\begin{lemma}
\label[lemma]{lem:arbitrary-density}
Let $\eps > 0$ and $\DG(V)$ be an $n$-vertex unit disk graph with 
diameter $D$. We can compute two sets
$R\subseteq Z\subseteq V$ with the following properties:
\begin{enumerate}[(i)]
\item $|R|\in O(D\eps^{-2})$ and $|Z|\in O(D\eps^{-3})$;
\item for every vertex $v\in V$, there is a 
  \emph{cluster vertex} $v' \in R$ with $d(v,v') \leq\eps$; and
\item for every $s, t\in R$, we have 
  $d(s,t)\leq d_Z(s,t)\leq(1+12\eps)d(s,t)+12\eps$, where $d_Z(s,t)$ 
  denotes the shortest path distance between $s$ and $t$ in $\DG(Z)$.
\end{enumerate}
\end{lemma}

\begin{lemma}
\label[lemma]{lem:diam-routing}
Let $\DG(V)$ be an $n$-vertex unit disk graph with diameter $D$. 
Furthermore, let $0<\eps\leq 1$.
There is a routing scheme with label size 
$O\left(\eps^{-3}D\log n\right)$ 
whose routing function
$\sigma_\textup{diam}$ achieves stretch factor $1+64\eps$.
\end{lemma}

\begin{figure}
\begin{center}
\includegraphics[scale=1]{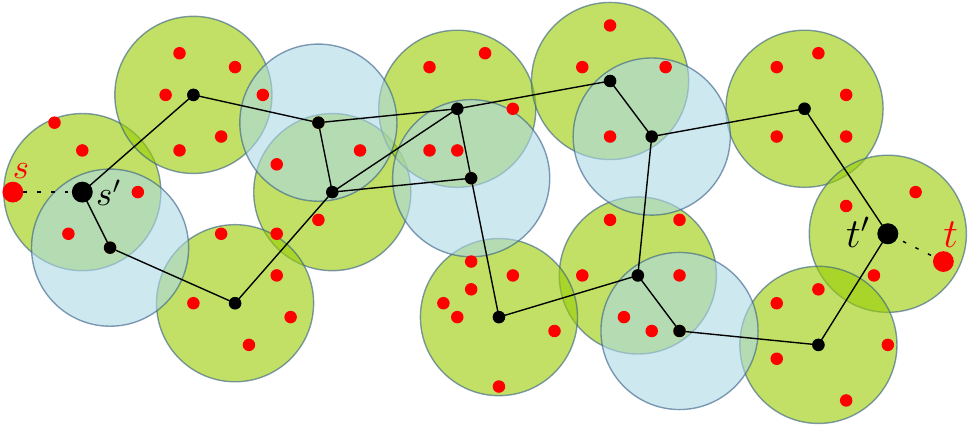}
\end{center}
\caption{The green disks represent the set $R$ from \cref{lem:arbitrary-density}.
The set $Z$ also contains the midpoints of the blue disks. The red and black dots represent the
whole input set $V$. To route from $s$ to $t$ we first take the hop to $s'$,
then route on the shortest path from $s'$ to $t'$ and finally take the hop to $t$.}
\label[figure]{fig:low-diameter-routing}
\end{figure}

\begin{proof}
The idea of the routing scheme is illustrated in \cref{fig:low-diameter-routing}.
First, we use \cref{lem:arbitrary-density} to find $R$ and $Z$.
Next, let $z\in R$. We use Dijkstra's algorithm
to compute a shortest path tree $T_z$ of the vertices $Z$ rooted at $z$.
We store the concatenation of $z_{\id}$ and
the pairs $(u_{\id}, v_{\id})$ in the label $\lab(z)$, for all edges $uv$ of 
$T_z$.
Next, let $v\in V\setminus R$. We pick an arbitrary cluster vertex
$v'\in R$ of $v$ with $d(v, v')\leq\eps$.
We store $v_{\id}$ and $\lab(v')$ in the label of $v$.
Finally, for every $v\in V$ we store a bit $b(v)$ in $\lab(v)$ that
is true if and only if $v\in R$.
Since $|Z| \in O(D\eps^{-3})$
by \cref{lem:arbitrary-density}, we immediately get 
$|\lab(v)| \in O\left(\eps^{-3}D\log\left(n\right)\right)$.

The routing function $\sigma_\text{diam}$ now works as follows: 
we are given $\lab(s)$, $\lab(t)$ and the broadcast-function
$\beta_s$. Using $\beta_s$, we check whether $s$ and $t$ are 
adjacent. If so, we use the port $\beta_s(t_{\id})$
to route the data packet. If not, we use the bit $b(v)$ to check
whether $s$ is a cluster vertex in $R$.
If $s$ is not a cluster vertex, $s$ and $s'$ are adjacent, 
since $\eps\leq 1$. We extract $s'_{\id}$ from $\lab(s')$ and route
the data packet via the port $\beta_s(s'_{\id})$ to $s'$. If $s$ is a cluster vertex,
$\sigma_\text{diam}$ extracts $t'_{\id }$ from $\lab(t)$, reconstructs
the shortest path tree $T_s$ for $s$, finds the next node $v$ on the unique path from
$s$ to $t'$ in $T_s$ and sends the data packet via the port $\beta_s(v_{\id})$ to $v$.

The correctness of the routing scheme is straightforward,
but we need to bound the stretch factor.
Let $s$ and $t$ be two vertices. If $d(s,t)\leq 1$,
they are neighbors and we route along a shortest path, 
i.e., $\delta(s,t)=d(s,t)$. Otherwise, let $d(s,t)> 1$ and
let $s'$ and $t'$ be their clusters ($s=s'$ and $t=t'$ is possible). 
Observe that $\delta(s',t')=d_Z(s',t')$. Hence, we can use
\cref{lem:arbitrary-density} to derive
\begin{align*}
\delta(s,t) &= |ss'| + \delta(s',t') + |t't| 
\leq \eps + d_Z(s',t') + \eps \leq (1+12\eps)d(s',t')+14\eps\\
&\leq (1+12\eps)(d(s,t)+2\eps)+14\eps = d(s,t)+12\eps d(s,t)+16\eps + 24\eps^2 \\
&\leq (1+28\eps+24\eps^2)d(s,t)\leq (1+52\eps)d(s,t).
\end{align*}
The last two inequalities hold because $d(s,t)>1$ and $\eps\leq 1$. 
Hence, $\delta(s,t)\leq(1+2^6\cdot\eps)d(s,t)$.\footnote{For notational 
ease, we restrict the explicit constants in our stretch bounds to  
powers of two.}
\end{proof}

\subsection{The Distance Oracle of Chan and Skrepetos}
\label{subsec:Chan-Skrepetos}

Our routing scheme is based on the recent approximate 
distance oracle for unit disk graphs by Chan and 
Skrepetos~\cite{ChanSk19}: we are given a
set $V \subset \R^2$ of $n$ points in the plane
and a parameter $\eps \geq D^{-1}$, where $D$ is the 
diameter of $\DG(V)$.
Chan and Skrepetos show how to compute in $O((1/\eps)^3 n \log^2 n)$
time a data structure of size $O((1/\eps)n \log n)$ that can 
answer \emph{approximate distance queries} in $\DG(V)$ in 
$O((1/\eps)\log n)$ time: given two vertices $s, t \in V$, compute 
a number $\theta \in \R$ with $d(s,t) \leq \theta  \leq d(s,t) + O(\eps D)$.
The main tool for this data structure is a suitable hierarchical 
decomposition of  $\DG(V)$.
More precisely, Chan and Skrepetos show that given $V$,
one can compute in $O(n \log n + (1/\eps) n)$ time a
\emph{decomposition tree} $\T$ for $\DG(V)$  with the
following properties.\footnote{The 
  reader familiar with the work of Chan and Skrepetos may 
  notice that we have slightly extended the notion of portals: 
  while Chan and Skrepetos define portals only for inner nodes, 
  we also define portals for the leaves. This does not change the 
essence of the decomposition, but makes the presentation more 
unified.}

\begin{itemize}
\item Every node $\mu$ of $\T$ is assigned two sets: 
  $\port(\mu)\subseteq V(\mu)\subseteq V$. The
  subgraph of $\DG(V)$ induced by $V(\mu)$ is connected and 
  the vertices in $\port(\mu)$ are called \emph{portals}.
\item If $\mu$ is the root, then $V(\mu)=V$.
\item If $\mu$ is an inner node with $k$ children 
  $\sigma_1,\dots,\sigma_k$, the sets 
  $\port(\mu),V(\sigma_1),\dots,V(\sigma_k)$ are pairwise disjoint,
  and we have $V(\sigma_i)\subseteq V(\mu)$, for $1\leq i\leq k$.
\item If $\mu$ is a leaf, then $V(\mu)=\port(\mu)$.
\item The height of $\T$ is in $O(\log n)$, and for every node 
  $\mu$ of $\T$, we have $|\port(\mu)| \in O(1/\eps)$.
\end{itemize}

To state the final (and most important) property of $\T$,
we first need to introduce some additional notation.
The properties of $\T$ so far imply that the portal sets of 
two different nodes in $\T$ are disjoint. For every portal $p$,
we let  $\mu(p)$ be the unique node in $\T$ with $p\in\port(\mu(p))$. 
Moreover, let $\mu$ be a node of $\T$ and 
$s, t \in V(\mu)$. We denote by $d_{\mu}(s,t)$  the shortest path
distance between $s$ and $t$ in the subgraph of $\DG(V)$ induced 
by $V(\mu)$. Now, the decomposition tree of Chan and Skrepetos 
has the property that for every pair of vertices $s, t \in V$,
if we set 
\[ 
  \theta(s,t) = \min_{\substack{p \text{ portal}\\s,t\in V(\mu(p))}}
  d_{\mu(p)}(s,p)+d_{\mu(p)}(p,t)
\]
then 
\begin{equation}
\label[equation]{equ:oracle-approx}
\theta(s,t)\leq d(s,t) + O(\eps D).
\end{equation}

\subsection{A Routing Scheme with Additive Stretch}
In the last section we presented a routing
scheme that is efficient for unit disk graphs with low diameter.
In this section we present a routing scheme that is efficient for 
unit disk graphs with large diameter.
Let $\DG(V)$ be an $n$-vertex unit disk graph with diameter $D$, 
and let $\eps>D^{-1}$. First, we set  $c=n\cdot(\eps D)^{-1}$ and 
define $x_c=\lfloor x\cdot c\rfloor$, for each $x\in\R^+_0$.
Next, we compute the 
decomposition tree $\T$, as explained
in \cref{subsec:Chan-Skrepetos}.

First, we describe the labels of the routing scheme. Let $v\in V$, 
and let $p$ be a portal with
$v\in V(\mu(p))$. We compute the shortest path tree $T_p$ of 
$V(\mu(p))$ rooted at $p$ and enumerate its
vertices in postorder. The postorder number of $v$ in $T_p$ is 
denoted by $r_p(v)$.
Next, the subtree of $T_p$ rooted at $v$ is called
$T_p(v)$ and we use $l_p(v)$ to denote the smallest postorder 
number in $T_p(v)$. Thus, since we enumerated the vertices in postorder,
a vertex $w\in V(\mu(p))$
is in the subtree $T_p(v)$ if and only if $r_p(w)\in[l_p(v), r_p(v)]$.
Finally, we apply the tree routing from \cref{lem:tree-routing} 
to $T_p$ and denote by $\lab_p(v)$ the corresponding label of $v$.
We store $(p_{\id} , d_{\mu(p)}(v,p)_c,l_p(v),r_p(v),\lab_p(v))$ in 
$\lab(v)$ and get the following lemma.

\begin{lemma}
\label[lemma]{lem:small-label}
For every vertex $v\in V$, 
we have $\vert\lab(v)\vert\in O\left(\dfrac{\log^3n}{\eps\log\log n}\right)$.
\end{lemma}
\begin{proof}
Since $\T$ has height $O(\log n)$, we know that $v$ is in 
$O(\log n)$ different sets $V(\mu)$. Moreover, for every node 
$\mu$, there are at most $O(1/\eps)$ portals. Thus, the label of 
$v$ contains $O(\eps^{-1}\cdot\log n)$ different entries. 
The value $d_{\mu(p)}(v,p)_c$ is a natural number, and
since $c\leq n$, we have 
\[ 
  d_{\mu(p)}(v,p)_c = \lfloor d_{\mu(p)}(v,p)\cdot c \rfloor\leq  n^2.
\]
Thus, we need $O(\log n)$ bits for the number $d_{\mu(p)}(v,p)_c$. 
Moreover, the identifier $p_{\id}$ as well as the postorder 
numbers stored in one entry only need $O(\log n)$ bits. Finally,
we apply \cref{lem:tree-routing} to conclude that one entry 
of the routing label has size $O(\log^2n/\log\log n)$.
The claim follows.
\end{proof}

\begin{figure}
\begin{center}
\includegraphics[scale=0.7]{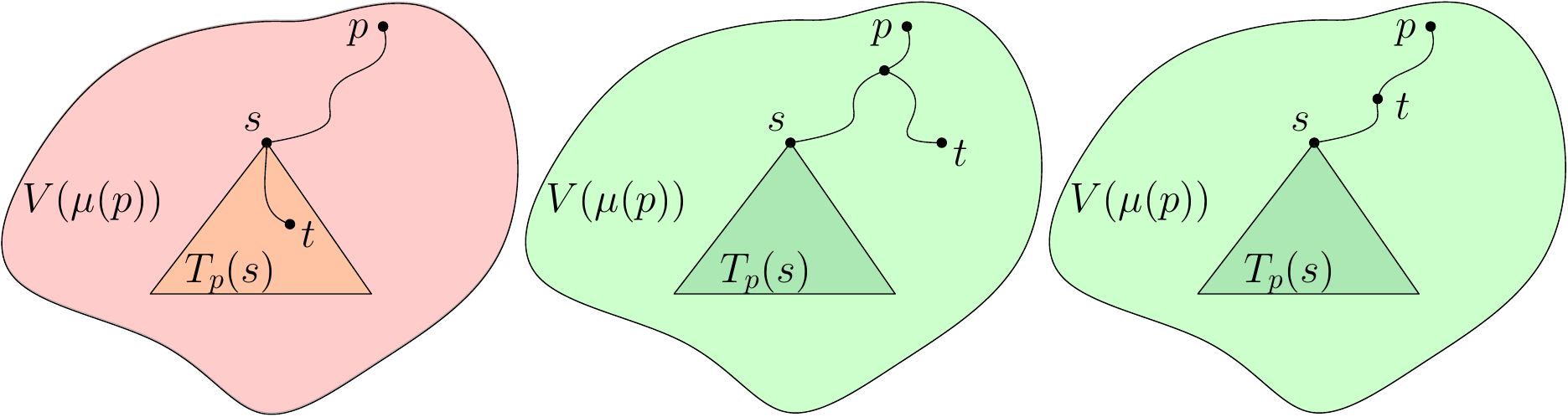}
\end{center}
\caption{Left: If $t$ is in $T_p(s)$, i.e., $\theta(s,t;p)=d_{\mu(p)}(t,p)-d_{\mu(p)}(p,s)$,
we route away from $p$. Middle and Right: If $t$ is not in $T_p(s)$, i.e.,
$\theta(s,t;p)=d_{\mu(p)}(t,p)+d_{\mu(p)}(p,s)$, we route towards $p$. The right picture
suggests to define $\theta(s,t;p)$ as $d_{\mu(p)}(s,p)-d_{\mu(p)}(t,p)$. This does
not influence the guarantees of our routing scheme but would lead to more cases.}
\label[figure]{fig:subtrees-cases}
\end{figure}

Next, we describe the routing function. We are given the labels 
$\lab(s)$ and $\lab(t)$ for the current vertex $s$ and the target 
vertex $t$. First, we identify all
portals $p$ with $s,t\in V(\mu(p))$. We can do this by identifying all vertices $p$
such that the entry $(p_{\id}, d_{\mu(p)}(s,p)_c,l_p(s),r_p(s),\lab_p(s))$ is in $\lab(s)$ and
the entry
$(p_{\id}, d_{\mu(p)}(t,p)_c,l_p(t),r_p(t),\lab_p(t))$ is in $\lab(t)$. Next, let
$\theta(s,t;p)=d_{\mu(p)}(t,p)+d_{\mu(p)}(p,s)$, if $t$ is not in the subtree
$T_p(s)$, and $\theta(s,t;p)=d_{\mu(p)}(t,p)-d_{\mu(p)}(p,s)$, otherwise;
see \cref{fig:subtrees-cases} for an illustration of the two cases.
Let $p_\text{opt}$ be the portal
that minimizes $\theta(s,t;p)$ among all portals $p$. Then, 
it is easy to see, that
$\theta(s,t;p_{opt})\leq \theta(s,t)$ (recall from \cref{subsec:Chan-Skrepetos} that 
$\theta(\cdot, \cdot)$ denotes the result of the distance 
oracle by Chan and Skrepetos). Hence, $\theta(s,t;p_{opt})$ is a good approximation for
the distance between $s$ and $t$. However, the routing function cannot compute
the optimal portal $p_\text{opt}$, since we do not have direct access to the 
real value $d_{\mu(p)}(s,p_\text{opt})$.
Instead, we use the values $d_{\mu(p)}(\cdot,p)_c$ to compute a near-optimal portal.
We define $\theta_c(s,t;p)=d_{\mu(p)}(t,p)_c+d_{\mu(p)}(p,s)_c$, if $t$ is not
in the subtree $T_p(s)$, and
$\theta_c(s,t;p)=d_{\mu(p)}(t,p)_c-d_{\mu(p)}(p,s)_c$, otherwise. Let $p_0$ be the
portal that lexicographically minimizes $(\theta_c(s,t;p),p_{\id})$, among all portals $p$.
We call $p_0$ the \emph{$s$-$t$-portal} and set $\theta_c(s,t)=\theta_c(s,t;p_0)$. Observe that
the $s$-$t$-portal can be computed by using only the labels of $s$ and $t$. The routing function now uses
the labels $\lab_{p_0}(s)$ and $\lab_{p_0}(t)$ as well as the broadcast function $\beta_s$
to compute the next vertex in $T_{p_0}$
and forwards the data packet to this vertex.

Finally, we have to show that the routing scheme is correct and routes along a short
(not necessarily shortest) path. For this, we first show that the routing process
terminates.

\begin{lemma}
\label[lemma]{lem:portal-remains-portal}
Let $s$ be the current vertex, $t$ the target vertex, 
and suppose that the routing scheme sends the packet 
from $s$ to $v$.
Moreover, let $p_0$ be the $s$-$t$-portal.
Then, $p_0$ is a possible candidate for the $v$-$t$-portal, and we have
$\theta_c(s,t;p_0)\geq\theta_c(v,t;p_0)+\vert sv\vert_c$.
\end{lemma}
\begin{proof}
First, let $\mu=\mu(p_0)$. Since $sv$ is an edge of the 
shortest path tree $T_{p_0}$, it follows that $v\in V(\mu(p_0))$.
This gives the first part of the claim. For the second part, 
we distinguish two cases:

\subparagraph*{Case 1: $t\in T_{p_0}(s)$.}  In this case, 
we have $t\in T_{p_0}(v)$, and thus
$\theta_c(v,t;p_0)=d_{\mu}(t,p_0)_c-d_{\mu}(p_0,v)_c$. 
Moreover, we have
\[
  d_{\mu}(p_0,v)_c = 
  \lfloor d_{\mu}(p_0,v)\cdot c\rfloor = 
  \lfloor d_{\mu}(p_0,s)\cdot c+\vert sv\vert\cdot c\rfloor
  \geq \lfloor d_{\mu}(p_0,s)\cdot c\rfloor+\lfloor\vert sv\vert\cdot c\rfloor 
  =d_{\mu}(p_0,s)_c+\vert sv\vert_c,
\]
since $s$ is on the path in $T_{p_0}$ from $p_0$ to $v$. Hence,
\[
  \theta_c(s,t;p_0) 
  = d_{\mu}(t,p_0)_c-d_{\mu}(p_0,s)_c
  \geq d_{\mu}(t,p_0)_c-d_{\mu}(p_0,v)_c +\vert sv\vert_c 
  = \theta_c(v,t;p_0) +\vert sv\vert_c,
\]
as desired.

\subparagraph*{Case 2: $t\notin T_{p_0}(s)$.} Similarly to 
the first case, we have
$d_{\mu}(p_0,s)_c\geq d_{\mu}(p_0,v)_c+\vert sv\vert_c$ and 
$\theta_c(v,t;p_0)\leq d_{\mu}(t,p_0)_c+d_{\mu}(p_0,v)_c$.
Thus, we get
\[
  \theta_c(s,t;p_0)
  =d_{\mu}(t,p_0)_c+d_{\mu}(p_0,s)_c\geq d_{\mu}(t,p_0)_c+d_{\mu}(p_0,v)_c +\vert sv\vert_c
\geq\theta_c(v,t;p_0) +\vert sv\vert_c, 
\]
and the claim follows.
\end{proof}

\begin{corollary}
\label[corollary]{cor:distance-doesnt-decrease}
Let $s$, $t$, and $v$ be as 
in \cref{lem:portal-remains-portal}.
Then, $\theta_c(s,t)\geq\theta_c(v,t)+\vert sv\vert_c$.
\end{corollary}

\begin{proof}
Let $p_0$ be the $s$-$t$-portal. From 
\cref{lem:portal-remains-portal}, we get
\[
  \theta_c(s,t)=
  \theta_c(s,t;p_0)
  \geq 
  \theta_c(v,t;p_0)+\vert sv\vert_c\geq\theta_c(v,t)+\vert sv\vert_c.
\]
The claim follows.
\end{proof}

\begin{lemma}
\label[lemma]{lem:portal-index-decreases}
Let $s$, $t$ and $v$ be as 
in \cref{lem:portal-remains-portal}.
Let $p$
be the $s$-$t$-portal and $q$ be the $v$-$t$-portal. Then, if
$\theta_c(s,t)=\theta_c(v,t)$, it follows that $p_{\id}\geq q_{\id}$.
\end{lemma}
\begin{proof}
From \cref{lem:portal-remains-portal}, we have
\[
  \theta_c(v,t;q)= \theta_c(v, t) = \theta_c(s, t) = 
\theta_c(s,t;p)\geq\theta_c(v,t;p)+\vert sv\vert_c
\geq\theta_c(v,t;p)\geq\theta(v,t;q).
\]
Hence, $\theta_c(v,t;p)=\theta_c(v,t;q)$. Furthermore, by 
construction, we have
$(\theta_c(v,t;p),p_{\id})\geq(\theta_c(v,t;q),q_{\id})$. 
Thus, the claim follows.
\end{proof}

\begin{lemma}
\label[lemma]{lem:routing-correct}
The routing scheme is correct.
\end{lemma}

\begin{proof}
Let $s$ be the current vertex and $t$ the 
desired target vertex, 
and let $p$ be the $s$-$t$-portal.
To measure the progress towards $t$, 
we consider the triple 
$(\theta_c(s,t),p_{\id},h_p(s,t))$, where 
$h_p(s, t)$ denotes the hop distance in $T_p$ 
between $s$ and $t$. i.e., the number of 
edges on the path between $s$ and $t$ in $T_p$.

Suppose that the routing 
scheme sends the packet from $s$ to $v$, and 
let $q$ be the $v$-$t$-portal. 
We argue that 
$(\theta_c(v,t),q_{\id},h_q(v,t)) <
(\theta_c(s,t),p_{\id},h_p(s,t)))$.
By \cref{cor:distance-doesnt-decrease}
and \cref{lem:portal-index-decreases}, 
it suffices to show that if
$\theta_c(s,t)=\theta_c(v,t)$
and $p = q$, then
$h_p(s,t)> h_q(v,t)$.
However, this is clear, because by \cref{lem:tree-routing}, 
$sv$ is an edge of $T_p$ that leads from $s$ towards $t$, and  
$T_q = T_p$.

Now, since the triples
$(\theta_c(s,t),p_{\id},h_p(s,t))$ lie in 
$\N^3$ and since $(0, 0, 0)$ is a global minimum, 
it follows that 
the data packet eventually
arrives at the target vertex $t$.
\end{proof}

\begin{lemma}
\label[lemma]{lem:bounding-stretch}
For any two vertices $s$ and $t$, we have $\delta(s,t)\leq d(s,t)+O(\eps D)$.
\end{lemma}

\begin{proof}
First, we show that $\theta_c(s,t)\leq c\cdot\theta(s,t)+1$: let $p_0$ be
the $s$-$t$-portal, and let $p_\text{opt}$ be the portal 
minimizing $\theta(s,t;\cdot)$
among all portals. Let $\mu=\mu(p_\text{opt})$. We obtain.
\begin{align*}
\theta_c(s,t)&=\theta_c(s,t;p_0)\leq\theta_c(s,t;p_\text{opt})
=\lfloor c\cdot d_{\mu}
(t,p_\text{opt})\rfloor\pm\lfloor c\cdot d_{\mu}(p_\text{opt},s)\rfloor\\
&\leq\lfloor c\cdot\left(d_{\mu}
(t,p_\text{opt})\pm d_{\mu}(p_\text{opt},s)\right)\rfloor +1
\leq\lfloor c\cdot\theta(s,t)\rfloor + 1 \leq c\cdot\theta(s,t)+1,
\end{align*}
where the $\pm$-operator is used to cover the two possible 
cases in the definition of $\theta_c$, and 
because $\lfloor a \rfloor + \lfloor b  \rfloor \leq 
\lfloor a + b \rfloor$ and 
$\lfloor a \rfloor - \lfloor b  \rfloor \leq 
\lfloor a - b \rfloor + 1$,  for all $a, b \geq 0$.
By \cref{lem:routing-correct}, we know that the routing terminates. 
Let $\pi: s=w_0,\dots,w_m=t$
be the routing path. From \cref{cor:distance-doesnt-decrease}, we 
get
$\vert w_i w_{i+1}\vert_c\leq \theta_c(w_i,t)-\theta_c(w_{i+1},t)$, 
and thus
\begin{align*}
\delta(s,t)&=\sum_{i=0}^{m-1}\vert w_i w_{i+1}\vert\leq\sum_{i=0}^{m-1}\dfrac{\vert w_i w_{i+1}\vert_c+1}{c}
=\dfrac{m}{c}+\dfrac{1}{c}\sum_{i=0}^{m-1}\vert w_i w_{i+1}\vert_c\\
&\leq\dfrac{m}{c}+\dfrac{1}{c}\sum_{i=0}^{m-1} (\theta_c(w_i,t)-\theta_c(w_{i+1},t))=\dfrac{m}{c}+\dfrac{\theta_c(s,t)}{c}\\
&\leq\dfrac{m}{c} + \dfrac{c\cdot\theta(s,t)+1}{c}=\dfrac{m+1}{c}+\theta(s,t)
\end{align*}
Now, using \cref{equ:oracle-approx} from \cref{subsec:Chan-Skrepetos}, the choice of 
$c = n \cdot (\eps D)^{-1}$, and the fact that $m\leq n-1$, 
we get 
\[
\dfrac{m+1}{c}+\theta(s,t)
\leq \frac{n}{n \cdot (\eps D)^{-1}} 
+ d(s, t) + O(\eps D) 
=  d(s, t) + O(\eps D),
\]
as claimed.
\end{proof}

We can now conclude with our first theorem.

\begin{theorem}
\label[theorem]{thm:additive-routing}
Let $\DG(V)$ be an $n$-vertex unit disk graph with diameter $D$. 
Furthermore, let $\eps > D^{-1}$.
There is a routing scheme with label size 
$O\left(\eps^{-1}\log^3 n/\log\log n\right)$ whose routing function
$\sigma_\textup{add}$ routes any data packet on a path with additive stretch $O(\eps D)$.
\end{theorem}

\section{A Routing Scheme with Stretch {\boldmath $1 + \eps$}}

Let $\DG(V)$ be an $n$-vertex unit disk graph with diameter $D$, 
and let $\eps > 0$. Furthermore, without loss of generality,
we can assume that $\eps\leq 1$. 
For our routing scheme, we need the following 
two ingredients from the literature.

\subparagraph*{Planar spanners.}
Let $c \geq 1$. A \emph{$c$-spanner} for $\DG(V)$ 
is a subgraph $H$ of $\DG(V)$ with vertex set $V$
such that for any $s, t \in V$, we have 
$d_H(s, t) \leq c \cdot d(s, t)$.
The following lemma shows  
the existence of good \emph{planar} spanners for unit 
disk graphs and was proven by Li, Calinescu, and Wan
\cite{li2002distributed}.

\begin{lemma}
\label[lemma]{lem:planar-spanner}
For any $n$-vertex unit disk graph $\DG(V)$, there exists 
a planar $4$-spanner $H \subseteq \DG(V)$. The spanner 
$H$ can be found in $O(n \log n)$ time.\footnote{Li, Calinescu, and
  Wan actually proved
that there is a planar $2.42$-spanner \cite{li2002distributed}. Since we do not care about 
the exact constant, we use a
power of $2$ to simplify later calculations.}
\end{lemma}

\subparagraph*{Sparse covers.}
Let $H =(V, E)$ be a weighted planar graph, and let $r \in \N$.
A \emph{sparse $r$-cover} for $H$ is a collection of 
connected subgraphs $H_1, H_2, \dots$ of $H$
with the following properties:
\begin{enumerate}[(i)]
\item for each vertex $v \in V$, there is at least one 
  subgraph $H_i$ that contains all the vertices $w \in V$
  with $d_H(v, w) \leq r$;
\item each vertex $v \in V$ is contained in $O(1)$ subgraphs $H_i$; 
  and
\item $\diam(H_i)\leq 2^6\cdot r$, for every subgraph $H_i$, 
  where $\diam(H_i)$ is the diameter of $H_i$.
\end{enumerate}
The following lemma establishes the existence of sparse covers 
for planar graphs and has been proven by Kawarabayashi, Sommer,
and Thorup \cite{kawarabayashi2013more}.

\begin{lemma}
\label[lemma]{lem:sparse-cover}
For any weighted planar graph $H$ with $n$ vertices and for any 
$r\in\N$, we can compute a sparse $r$-cover for $H$ in $O(n \log n)$ 
time.\footnote{Actually, it is possible to prove an upper bound of 
  $48r$ on the diameters of the subgraphs \cite{kawarabayashi2013more}, 
  but we again prefer a power of two in order to 
  simplify subsequent calculations.}
\end{lemma}

\subparagraph*{The Routing Scheme.}
Now we have all ingredients for our final routing scheme. 
In the preprocessing phase, we compute a planar $4$-spanner $H$
of $\DG(V)$, as in \cref{lem:planar-spanner}. Then, we have
$\diam(H)\leq 4D$. 
Next, for each
$k\in\I=\{\lceil\log\frac{8}{\eps}\rceil,\dots, \lceil\log(4D)\rceil\}$, 
we use \cref{lem:sparse-cover} to
construct a sparse $2^k$-cover $(H_1^k,H_2^k,\dots)$ of $H$.
Let $G_i^k$ be the induced unit disk graph on the vertex set of 
$H_i^k$. Let $k_0=\lceil\log\frac{8}{\eps}\rceil$, for each $G_i^{k_0}$,
we apply the preprocessing mechanism of the low diameter routing 
scheme from \cref{lem:diam-routing}. For each $k\in\I\setminus\{k_0\}$, 
we apply to each $G_i^k$ the preprocessing step of
the routing scheme with additive stretch 
from \cref{thm:additive-routing}. We use $\lab_{k,i}$ to denote the 
resulting labeling for the graph $G_i^k$, for $k\in\I$.

Now, we describe how to obtain the labels for our 
routing scheme. Let $v$ be a vertex of $\DG(V)$ and let 
$k\in\I$. 
Since $H_1^k, H_2^k, \dots$ is a sparse $2^k$-cover, 
there exists an index $i(v,k)$ such that 
$H_{i(v,k)}^k$ contains all vertices 
$w \in V$ with $d_H(v, w) \leq 2^k$.
Now, for each $v \in V$, the label 
$\lab(v)$ is the concatenation of the 
tuples
$\big(k, i, b(i,k,v), \lab_{k,i}(v)\big)$,  
for each $k\in\I$ and each $i$ with
$v\in V(G_i^k)$.
Here $b(i,k,v)$ 
is a Boolean value that is true if
and only if $i=i(v,k)$. The following lemma 
bounds the maximum label size.
\begin{lemma}
\label[lemma]{lem:lab-size}
For every vertex $v\in V$, we have
$|\lab(v)|\in O\left(\dfrac{\log D\log^3n}{\eps\log\log n}+\dfrac{\log\left(n\right)}{\eps^4}\right)$.
\end{lemma}
\begin{proof}
Since there are $O(\log D)$ different values for $k$, 
and since for each $k$, the vertex $v$ appears in 
$O(1)$ subgraphs $G_i^k$, we have that 
$v$ lies in $O(\log D)$ different subgraphs
$G_i^k$. For the subgraphs $G_i^{k_0}$, the label 
$\lab_{k_0, i}(v)$ comes from the low diameter routing 
scheme.  Since
$\diam(G_i^{k_0})\in O(1/\eps)$, \cref{lem:diam-routing} implies 
that $\lab_{k_0, i}(v)$ needs $O(\log (n)/\eps^4)$ bits.
Since $v$ lies in $O(1)$ subgraphs $G_i^{k_0}$, we can conclude that 
the corresponding tuples in $\lab(v)$ require 
$O(\log (n)/\eps^4)$ bits in total.
For the remaining $O(\log D)$ subgraphs, we derive the 
label $\lab_{k, i}(v)$ from the
additive stretch routing scheme from 
\cref{thm:additive-routing}. Hence, the corresponding tuples 
take $O(\eps^{-1}\log D\log^3 n/\log\log n)$ bits in total. 
The claim follows.
\end{proof}

We next describe the routing function $\sigma$, see \cref{fig:onion-of-t}.
Suppose we are given the labels $\lab(s)$ and $\lab(t)$ 
of the current vertex $s$ and the target $t$, 
together with the broadcast function $\beta_s$. The routing function
works as follows: we find the smallest number $k=k(s,t)\in\I$ such 
that there is an index $i$ for which
the tuple $(k,i,\true,*)$ is in
$\lab(t)$ and the tuple $(k,i,*, *)$ 
is in $\lab(s)$.\footnote{The $*$ is a placeholder for an arbitrary 
value. Note that $\lab(s)$ and $\lab(t)$ each contain 
at most one tuple that starts with $k, i$}
We can now derive
the following observation:

\begin{observation}
\label[observation]{obs:onion-of-t}
Let $s,t$ be vertices of $G_i^k$ with $k=k(s,t)$. Then we have
$d(s,t)\leq 2^{k+6}$. Moreover, if $k> k_0$ we have $d(s,t)\geq 2^{k-3}$.
\end{observation}
\begin{proof}
By property (iii) of a sparse cover we get
$d(s,t)\leq \diam(G_i^k)\leq\diam(H_i^k)\leq2^{k+6}$. This proves the first
inequality.

Next, let $k > k_0$. The minimality of $k$ and 
property (i) of a sparse cover show that $d_H(s,t)\geq 2^{k-1}$
Finally, since $H$ is a $4$-spanner of $G$ we derive $d(s,t)\geq2^{k-3}$
and the claim follows.
\end{proof}
Once we have $k$ and $i$, we can route in $G_i^k$ using the 
labels $\lab_{k,i}(s)$ and $\lab_{k,i}(t)$ as well as the
broadcast function $\beta_s$. If $k=k_0$, we use $\sigma_\text{diam}$,
otherwise, we use $\sigma_\text{add}$ to compute the correct port.
It remains to show the correctness and to analyze the stretch factor. We start with the correctness.
Its proof is quite similar to the correctness proof of
$\sigma_\text{add}$.

\begin{figure}
\begin{center}
\includegraphics[scale=1]{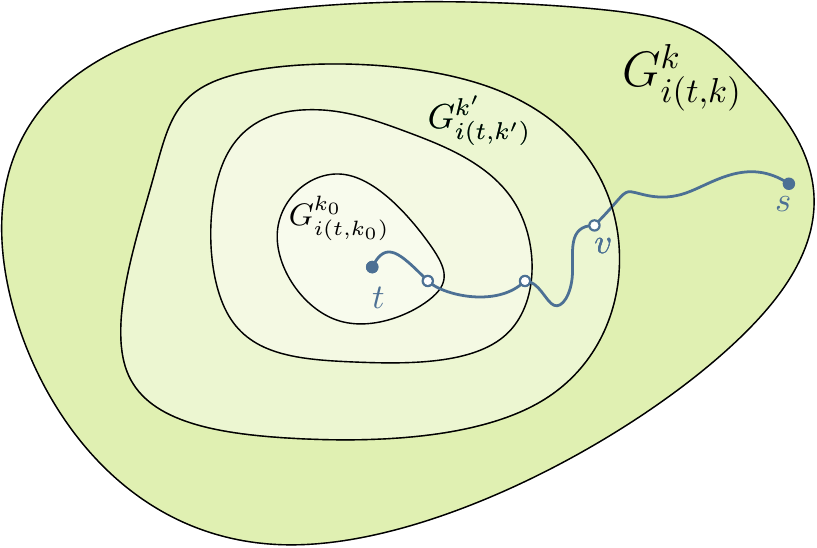}
\end{center}
\caption{It is $2^{k-3}\leq\diam\left(G_{i(t,k)}^k\right)\leq 2^{k+5}$.
We use the additive stretch routing scheme to route within
$G_{i(t,k)}^k$ until we find a vertex $v$ that is in $G_{i(t,k')}^{k'}$
for $k'<k$. This process continues until we find a vertex that is in $G_{i(t,k_0)}^{k_0}$,
here we use the low diameter routing scheme until we reach $t$.}
\label[figure]{fig:onion-of-t}
\end{figure}

\begin{lemma}
\label[lemma]{lem:correctness}
The routing scheme is correct.
\end{lemma}
\begin{proof}
Let $s$ be the current vertex, $t$ the desired target 
vertex, and suppose that the routing scheme sends the 
packet to vertex $v$ from $s$. Moreover, let $k=k(s,t)$ 
and $i=i(s,t)$ be two indices that were used by the 
routing function to determine $v$.
Since the routing step from $s$ to $v$ 
takes place in the graph $G_i^k$, we know that $k$
is a potential candidate for $k(v,t)$. 
Thus, $k(v,t) \leq k$. If $k(v,t) < k$, we have made progress.
However, if $k(v,t) = k$, it must be that $i(s,t)=i(v,t)$, since 
we defined $\lab(t)$ such that
for each $k$, there is exactly one $i$ with $b(i,k,t)=\true$.
This means that if $k$ does not change, the routing 
continues in the subgraph $G_i^k$. We already
proved in \cref{lem:diam-routing} and \cref{lem:routing-correct}
that the underlying routing scheme for this task is 
correct. Hence, after a 
finite number of steps, we either reach $t$, or we decrease the 
value $k$. Since there is only a finite number of values for $k$,
correctness follows.
\end{proof}

The next lemma bounds the additive stretch as a function 
of $k$.
\begin{lemma}
\label[lemma]{lem:stretch}
There is a constant $c > 0$ with the following property:
let $s$ and $t$ be two vertices and let $k=k(s,t)$. Then, we have
$\delta(s,t)\leq d(s,t)+ c\eps \cdot 2^k$.
\end{lemma}
\begin{proof}
We  use induction on $k \geq k_0$. 
First, suppose that 
$k = k_0 = \lceil\log(8/\eps)\rceil$ and let 
$s, t$ be two vertices with $k(s, t) = k_0$. 
Let $G_i^k$ be
the graph that is used to determine the next vertex after $s$.
Since $k$ can only decrease during the routing, 
and since $k_0$ is the minimum possible value of $k$, we route 
within $G_i^k$, using
the low diameter routing scheme, until we reach
$t$. Moreover, by \cref{lem:diam-routing} and \cref{obs:onion-of-t}, 
and for $c\geq 2^{12}$ we get
\[
  \delta(s,t)\leq(1+\eps\cdot 2^6)d(s,t)\leq d(s,t) + \eps\cdot 2^{k+12}\leq d(s,t)+c\eps\cdot 2^k.
\]
Next, assume that $k > k_0$. Let $s, t$ be two 
vertices with $k(s,t) = k$, and assume that for every vertex $w$ with 
$k(w,t) < k$, we have $\delta(w,t)\leq d(w,t)+c\eps\cdot 2^{k(w,t)}$.
Let $G_i^k$ be the graph in which our scheme chooses to 
route the data packet from $s$ to the next node. Let $v$ be the 
first node
on the routing path from $s$ to $t$ for which $k(v,t)<k$, see \cref{fig:onion-of-t}. 
Moreover, let $\delta'(\cdot, \cdot)$ measure the length of the 
routing path within the subgraph $G_i^k$, using the additive stretch 
routing scheme. 
Next, by the definition of $k_0$ and since $k>k_0$ we get
$\diam(G_i^k)\geq d(s,t)\geq 2^{k-3}\geq1/\eps$ from
\cref{obs:onion-of-t}.
Furthermore, we know that $d(v,t)\leq\delta'(v,t)$, since $t$ is a 
vertex in $G_i^k$. Finally, we use the inductive hypothesis
as well as \cref{thm:additive-routing} to derive
\begin{align*}
\delta(s,t) &= \delta'(s,v)+\delta(v,t)\leq \delta'(s,v)+d(v,t)+
c \eps\cdot 2^{k(v,t)}\leq\delta'(s,v)+\delta'(v,t)+ c\eps\cdot 2^{k-1}\\
&=\delta'(s,t)+c\eps\cdot 2^{k-1}\leq d(s,t) +c_0\eps\cdot2^{k+5} +c \eps\cdot 2^{k-1}\leq d(s,t) + 
c\eps\cdot 2^k,
\end{align*}
for $c\geq c_02^6$, where $c_0$ is the constant from the $O$-Notation
of the stretch in \cref{thm:additive-routing}.  Hence, the claim follows.
\end{proof}
Finally, we can put everything together to obtain our main 
theorem.
\begin{theorem}
Let $\DG(V)$ be an $n$-vertex unit disk graph and $D$ its diameter. Furthermore, let $\eps > 0$.
There is a routing scheme with $O\left(\log D\log^3 n/\log\log n\right)$ label 
size\footnote{The constant in the O-Notation depends on $\eps$.} whose routing function
$\sigma$ achieves the stretch factor $1+\eps$.
\end{theorem}

\begin{proof}
It remains to show the stretch factor. Here, it suffices to show that the stretch factor is $1+O(\eps)$.
Let $s$ and $t$ be two vertices and $k=k(s,t)$.
If $k=k_0$ the stretch factor immediately follows from \cref{lem:diam-routing}. Thus, assume
$k\neq k_0$. On the one hand
we know from \cref{obs:onion-of-t} that $2^{k-3}\leq d(s,t)$, and on the other hand we know from \cref{lem:stretch} that
$\delta(s,t)\leq d(s,t)+c\eps\cdot 2^k$. Plugging everything together we get
\begin{align*}
\delta(s,t)\leq d(s,t)+c\eps\cdot 2^k\leq d(s,t)+c2^3\eps\cdot d(s,t)=(1+c2^3\eps)d(s,t).
\end{align*}
This gives the desired stretch factor and the theorem follows.
\end{proof}

\section{Conclusion}

We presented the first efficient, compact, and headerless routing scheme for unit disk graphs. It achieves
near-optimal stretch $1+\eps$ and uses $O(\log D\log^3n/\log\log n)$ bits in the label. 
It would be interesting to see if this result
can be extended to disk graphs in general. If the radii of the disks are
unbounded, the decomposition of
Chan and Skrepetos cannot be applied immediately.
However, the case of bounded radii is still 
interesting, and even there, it is not clear 
how the method by Chan and Skrepetos generalizes.

Finally, let us compare our routing scheme to the known schemes. The model of the routing scheme
of Kaplan et al.~\cite{KaplanMuRoSe18} is very close to ours. The routing scheme can be implemented
using the fixed-port model. Moreover, they also use some kind of broadcasting function, since they
claim that neighborhood can be checked locally. The scheme was generalized to non-unit disk graphs with constant bounded
radii~\cite{willert2016routing}. Nevertheless, in unit 
disk graphs,
we achieve the same stretch factor and still have additional 
information of poly-logarithmic size.
The main advantage of our routing scheme is that we do not use any additional headers.
Therefore, whenever a data packet arrives at a node, it is not necessary to know what happened before
or where the packet came from. In the routing scheme of Kaplan et al., it happens that a data
packet visits a node more than once.

The routing scheme of Yan et al. \cite{yan2012compact} uses headers as well, but they are only computed in the first step
and do not change again. The idea
of their routing scheme is similar to ours: the graph is covered by $O(\log n)$ different trees. When
the routing starts, the labels of the source and the target are used to determine the identity of a
tree and an $O(\log n)$-bit label of the target within this tree. Finally, they completely forget the
original labels and route within this tree until they reach $t$. Their stretch is bounded by a constant.
Our routing scheme can also be turned into this model, but we have $O(\log D\log n)$ different trees
that cover the unit disk graph and the label of a vertex in one of the trees has size $O(\log^2n/\log\log n)$.
Nevertheless, we achieve the near optimal stretch $1+\eps$. Moreover, Yan et al. use the designer-port model
and thus, they can route within a tree using labels of size $O(\log n)$. But since nodes are contained in
more than one tree, there have to be lookup-tables for the port assignments. Their routing scheme can easily
be turned into the fixed-port model: the stretch would not change and the label size would increase to
$O(\log^3n/\log\log n)$. In conclusion, our routing scheme needs an $O(\log D)$-factor more in the label
size but achieves near-optimal stretch $1+\eps$ and the 
underlying routing model is specified more clearly.

\bibliographystyle{plainurl}
\bibliography{sources.bib}

\end{document}